\documentclass{article}
\usepackage{amsmath,amsthm,times}
\usepackage{latexsym}
\usepackage{bbm}

\newcommand{\fff}{f:\bit^n \times \bit^n \times \bit^n \ra \bit}
\newcommand{\ffff}{f:(\bit^n)^k \ra \bit}
\newcommand{\ggg}{g:\bit^n \times \bit^n \ra \bit}
\newcommand{\F}{\mathbbm{F}}

\newcommand{\select}{\sigma}

\newcommand{\Nat}{\mathbbm{N}} %natural numbers
\newcommand{\Z}{\mathbbm{Z}} %natural numbers
\newcommand{\dmax}{d_\mathrm{max}}

\newcommand{\remove}[1]{}

\newcommand{\myparagraph}[1]{\smallskip \noindent {\bf #1}}

\newcommand{\half}{\frac{1}{2}} 
\newcommand{\lset}{\left \{} \newcommand{\rset}{\right \}}
\newcommand{\pair}[1]{\langle #1 \rangle}
\newcommand{\set}[1]{\lset #1 \rset}   \newcommand{\bit}{\{0,1\}}

\newcommand{\eqndef}{\stackrel{\mbox{\tiny def}}{{=}}}

\newcommand{\conc}{\mbox{\scriptsize $\circ$}}
\newcommand{\eqdef}{\eqndef}

\newcommand{\calF}{\F}

\newcommand{\calR}{{\cal R}}

\newcommand{\ra}{\rightarrow}

\newcommand{\se}{{\subseteq}}

\newcommand{\rank}{\operatorname{rank}}
\newcommand{\row}{{\rm row}}

\newcommand{\poly}{{\rm poly}}

\newcommand{\email}[1]{E-mail: #1}%%%{{\tt #1}}
\newcommand{\fax}[1]{}

\newtheorem{theorem}{Theorem}[section]
\newtheorem{lemma}[theorem]{Lemma}
\newtheorem{definition}[theorem]{Definition}

\newtheorem{corollary}[theorem]{Corollary}
\newtheorem{proposition}[theorem]{Proposition}

\newtheorem{Remark}[theorem]{Remark}

\newenvironment{remark}{\begin{Remark}\begin{rm}}{\end{rm}\end{Remark}}
\newtheorem{open}{Question}
\newtheorem{conjecture}{Conjecture}
\newtheorem{observation}[theorem]{Observation}
\newenvironment{proofapp}[1]{\noindent {\bf Proof \ #1:} \hspace{.677em}}%
  {\qed}

\newcommand{\lan}{\langle}
\newcommand{\ran}{\rangle}
\newcommand{\mC}{\mathcal{C}}

\begin{document}

\title{Partition Arguments in Multiparty Communication Complexity}

\author{Jan Draisma\thanks{Department of Mathematics and Computer Science,
Technische Universiteit Eindhoven; and CWI Amsterdam.
\email{j.draisma@tue.nl}. Partially supported by MSRI (Berkeley) and by NWO mathematics cluster
DIAMANT.} \and Eyal Kushilevitz\thanks{Computer Science
Department, Technion Israel Institute of Technology, Haifa.
\email{eyalk@cs.technion.ac.il}. Research supported by grant
1310/06 from the Israel Science Foundation (ISF).} \and Enav
Weinreb\thanks{Computer Science Department, Technion Israel
Institute of Technology, Haifa. \email{weinreb@cs.technion.ac.il}.
Research supported by grant 1310/06 from the Israel Science
Foundation (ISF).} }

%\institute{}
\maketitle

\begin{abstract}
Consider the ``Number in Hand'' multiparty communication
complexity model, where $k$ players holding
inputs $x_1,\ldots,x_k\in\bit^n$ communicate
to compute the value $f(x_1,\ldots,x_k)$ of a function $f$ known to all of them. The main lower
bound technique for the communication complexity of such problems
is that of {\em partition arguments}: partition the $k$ players
into two disjoint sets of players and find a lower bound for the
induced two-party communication complexity problem.

In this paper, we study the power of partition arguments. Our two
main results are very different in nature: \\
(i) For {\em randomized} communication complexity, we show that
partition arguments may yield bounds that are exponentially far
from the true communication complexity. Specifically, we prove
that there exists a $3$-argument function $f$ whose communication
complexity is $\Omega(n)$, while partition arguments can only
yield an $\Omega(\log n)$ lower bound. The same holds for {\em
nondeterministic}
communication complexity. \\
(ii) For
%% the case of
{\em deterministic} communication complexity, we prove that
finding significant gaps between the true communication complexity
and the best lower bound that can be obtained via partition
arguments, would imply progress on a generalized version of the
``log-rank conjecture'' in communication complexity.
%This explains why partition arguments
%seem to yield good lower bounds in the deterministic case.

We conclude with two results on the multiparty ``fooling set technique'',
another method for obtaining communication complexity lower bounds.
\end{abstract}

\section{Introduction}
\label{s:intro} Yao's two-party communication
complexity~\cite{KN97,Yao} is a well-studied model, of which
several extensions to multiparty settings were considered in the
literature. In this paper, we consider the following extension
that is arguably the simplest one (alternative multiparty models
are discussed below): there are $k>2$ players, $P_1,\ldots,P_k$,
where each player $P_i$ holds an input $x_i\in\bit^n$. The players
communicate by using a broadcast channel (sometimes referred to as
a ``blackboard'' in the communication complexity literature) and
their goal is to compute some function $f$ evaluated at their
inputs, i.e., the value $f(x_1,\ldots,x_k)$, while minimizing the
number of bits communicated.\footnote{If broadcast is not
available, but rather the players are connected via point-to-point
channels, this influences the communication complexity by a factor
of at most $k$; we will mostly view $k$ as a constant (e.g.,
$k=3$) and hence the difference is minor.}

As in the two-party case, the most interesting question for such a
model is proving lower bounds, with an emphasis on ``generic''
methods. The main lower bound method known for the above
multiparty model is the so-called {\em partition argument} method.
Namely, the $k$ players are partitioned into two
%%  non-empty
disjoint sets of players, $A$ and $B$, and we look at the induced
two-argument function $f^{A,B}$ defined by $f^{A,B}(\{x_i\}_{i\in
A},\{x_i\}_{i\in B})\eqdef f(x_1,\ldots,x_k)$. Then, by applying
any of the various lower-bound methods known for the two-party
case, we obtain some lower bound $\ell_{A,B}$ on the (two-party)
communication complexity of $f^{A,B}$. This value is obviously a
lower bound also for the (multiparty) communication complexity of
$f$. Finally, the partition arguments bound $\ell_{\rm PAR}$ is
the best lower bound that can be obtained in this way; namely,
$\ell_{\rm PAR}=\max_{A,B}\set{\ell_{A,B}}$, where the maximum is
taken over all possible partitions $A,B$ as above.

The fundamental question studied in this paper is whether partition
arguments suffice for determining the multiparty communication
complexity of every $k$-argument function $f$; or, put differently,
how close the partition argument bound is to the true communication
complexity of $f$. More specifically,

\myparagraph{Question:} Is there a constant $c \ge 1$ such that,
for every $k$-argument function $f$, the $k$-party communication
complexity of $f$ is between $\ell_{\rm PAR}$ and $(\ell_{\rm
PAR})^c$~?

\smallskip

As usual, this question can be studied with respect to various
communication complexity models (deterministic, non-deterministic,
randomized etc.). If the answer to this question is positive, we
will say that partition arguments are {\em universal} in the
corresponding model.

\myparagraph{Our Results:} On the one hand, for the deterministic case
(Section~\ref{s:det}), we
explain the current state of affairs where partition arguments seem to
yield essentially the best known lower bounds.  We do this by relating
the above question, in the deterministic setting, to one of the central
open problems in the study of communication complexity, the so-called
``log-rank conjecture'' (see~\cite{RS95, NW95a} and the references
therein), stating that the deterministic communication complexity of every
two-argument boolean function $g$ is polynomially-related to the log of
the algebraic rank (over the reals) of the matrix $M_g$ corresponding to
the function. Specifically, we show that if a natural generalization of
the log-rank conjecture to $k$ players holds then the answer to the above
question is positive; namely, in this case, the partition arguments bound
is polynomially close to the true multiparty communication complexity. In
other words, a negative answer to the above question implies refuting
the generalized log-rank conjecture. Furthermore, we characterize
the collections of partitions one has to consider in order to decide if
the rank lower bound is applicable for a given $k$-argument function.
Specifically, these are the collections of partitions such that for
every two players $P_i$ and $P_j$ there is a partition $A, B$ such that
$i \in A$ and $j \in B$. That is, if all induced two-argument functions
in such a collection are easy, then, assuming the generalized log-rank
conjecture, the original function is easy as well.

On the other hand, we show that both in the case of non-deterministic
communication complexity (Subsection~\ref{ssec:nd}) and in the case of
randomized
communication complexity (Subsection~\ref{ssec:rnd}), the answer to the above question is
negative in a strong sense. Namely, there exists a $3$-argument
function $f$, for which \remove{the partition-argument bound cannot yield
a lower bound better than $\Omega(\log n)$ (or, in other words,}
each of the induced two-party functions has an upper bound of
$O(\log n)$, while the true $3$-party communication complexity
of $f$ is exponentially larger, i.e. $\Omega(n)$. Of course, other
methods than partition arguments are needed here to prove the
lower bound on the complexity of $f$. Specifically, we pick $f$ at
random from a carefully designed family of functions, where the
induced two-argument functions for all of them have low
complexity, and show that with positive probability we will get a
function with large multiparty communication
complexity.\footnote{Note that, in order to give a negative answer
to the above question, it is enough to discuss the case $k=3$.
This immediately yields a gap also for larger values of $k$.} We
also show that, in contrast to the situation with respect to the
deterministic communication complexity of {\em functions} (as
described above), there exist $k$-party {\em search problems
(relations)} whose deterministic communication complexity is
$\Omega(n)$ while all their induced relations can be solved
without communicating at all (Subsection~\ref{ssec:rel}).

We accompany the above main results by two additional results on
the so-called ``fooling set technique'' in the multiparty case
(Section~\ref{a:fs}). First, we prove the existence of a
$3$-argument function $f$ for which there exists a large fooling
set that implies an $\Omega(n)$ lower bound on the deterministic
communication complexity of $f$, but where all the induced
two-party functions have only very small fooling sets. However,
extending results from~\cite{DHS96} for the two-party case, we
prove that lower bounds on the communication complexity of a
$k$-argument function obtained with the fooling set technique
cannot be significantly better than those obtained with the rank
lower bound.

\myparagraph{Related work:} Multiparty communication complexity
was studied in other models as well. Dolev and
Feder~\cite{DF92,DF89} (see also~\cite{Duris04, DR98}) studied a
$k$-party model where the communication is managed via an
additional party referred to as the ``coordinator''. Their main
result is a proof that the maximal gap between the deterministic
and the non-deterministic communication complexity of every
function is quadratic even in this multiparty setting. Their
motivation was bridging between the two-party communication
complexity model and the model of decision trees, where both have
such quadratic gaps. Our model differs from theirs in terms of the
communication among players and in that we concentrate on the case
of a small number of players.

Another popular model in the study of multiparty communication
complexity is the so-called ``Number On the Forehead'' (NOF)
model~\cite{CFL,BNS}, where each party $P_i$ gets all the inputs
$x_1,\ldots,x_k$ {\em except} for $x_i$. This model is less
natural in distributed systems settings but it has a wide variety
of other applications. Note that in the NOF
model, partition arguments are useless because any two players
when put together know the entire input to $f$.

Our results concern the ``Number in Hand'' $k$-party model. Lower
bound techniques different from partition arguments were presented
by Chakrabarti et al.~\cite{CKS03}, following~\cite{AMS99,
BJKS04}. These lower bounds are for the ``disjointness with unique
intersection'' {\em promise problem}. In this problem, the $k$
inputs are subsets of a universe of size $n$, together with the
promise that the $k$ sets are either pairwise disjoint, in which
case the output is $0$, or uniquely intersecting, i.e. they have
one element in common but are otherwise disjoint, in which case
the output is $1$. Note that partition arguments are useless for
this promise problem: any two inputs determine the output.
Chakrabarti et al. prove a near optimal lower bound of $\Omega(n /
k \log k)$ for this function, using information theoretical tools
from~\cite{BJKS04}. Their result is improved to the optimal lower
bound of $\Omega(n/k)$ in \cite{Gronemeier09}. This problem has
applications to the space complexity of approximating frequency
moments in the data stream model (see~\cite{AJKS02, AMS99}). As
mentioned, we provide additional examples where partition
arguments fail to give good lower bounds for the deterministic
communication complexity of {\em relations}. It should be noted,
however, that there are several contexts where the communication
complexity of relations and, in particular, of promise problems,
seems to behave differently than that of {\em functions} (e.g, the
context of the ``direct-sum'' problem~\cite{FKNN95}). Indeed, for
functions, no generic lower bound technique different than
partition arguments is known.
%% (as opposed to relations and promise problems).

\remove{\myparagraph{Organization.} In Section~\ref{s:prel}, we
discuss some required background in communication complexity.
Then, in Section~\ref{s:det}, we discuss the power of partition
arguments for deterministic communication complexity and show a
connection between the universality of the partition arguments
technique and the log-rank conjecture in communication
complexity. In particular, in Section~\ref{ss:3} we prove the
above connection for three-party communication complexity using
elementary linear algebraic tools, and in Section~\ref{ss:k} we
extend the connection to the $k$-party communication complexity
case, where tools from tensor algebra are used. Then, in
Section~\ref{s:nd}, we show that for the randomized and
non-deterministic models, there are examples of functions in which
partition arguments are extremely weak. We note that
Section~\ref{s:nd} does not rely on Section~\ref{s:det} (hence,
the reader may choose to skip parts of the proof). \remove{Finally, in
Appendix~\ref{a:fs}, we discuss relations between the ``fooling
set'' lower bound technique in communication complexity and
partition arguments.}
}

\section{Preliminaries}
\label{s:prel}
\myparagraph{Notation.}
For a positive integer $m$, we denote by $[m]$ the set
$\set{1,2,\dots,m}$. All the logarithms in this paper are to the
base $2$.
%%  by default, unless stated otherwise.
For two strings $x,y \in \bit^*$, we use $x \conc y$ to denote
their concatenation. We refer by $poly(n)$ to the set of functions
that are asymptotically bounded by a polynomial in $n$.

\myparagraph{Two-Party Communication Complexity.}
For a Boolean function $\ggg$, denote by $D(g)$ the deterministic
communication complexity of $g$, i.e., the number of bits Alice,
holding $x \in \bit^n$, and Bob, holding $y \in \bit^n$, need to
exchange in order to jointly compute $g(x,y)$.  Denote by $M_g \in
\bit^{2^n \times 2^n}$ the matrix representing $g$, i.e.,
$M_g[x,y] = g(x,y)$ for every $(x,y) \in \bit^n \times \bit^n$.

\myparagraph{$k$-Party Communication Complexity.}
Let $\ffff$ be a Boolean function. A set of $k$ players $P_1,\dots,P_k$
hold inputs $x_1,\dots,x_k$ respectively, and wish to compute
$f(x_1,\dots,x_k)$. The means of communication is broadcast. Again,
we denote by $D(f)$ the complexity of the best deterministic protocol
for computing $f$ in this model, where the complexity of a protocol
is the number of bits sent on the worst-case input. Generalizing the
two-argument case, we represent $f$ using a $k$-dimensional tensor $M_f$.
For any partition $A,B$ of $[k]$ we denote by $f^{A,B}$ the induced
two-argument function.

\myparagraph{Non-Deterministic Communication Complexity.} For $b
\in \bit$, a $b$-monochromatic (combinatorial) rectangle of a
function $\ggg$ is a set of pairs of the form $X \times Y$, where
$X, Y \se \bit^n$, such that for every $x \in X$ and $y \in Y$ we
have that $g(x,y)=b$. A $b$-cover of $g$ of size $t$ is a set of
(possibly overlapping) $b$-monochromatic rectangles $\calR =
\set{R_1, \dots, R_t}$ such that, for every pair $(x,y) \in \bit^n
\times \bit^n$, if $g(x,y)=b$ then there exists an index $i \in
[t]$ such that $(x,y) \in R_i$. Denote by $C^b(g)$ the size of the
smallest $b$-cover of $g$. The non-deterministic communication
complexity of $g$ is denoted by $N^1(g) = \log C^1(g)$. Similarly,
the co-non-deterministic communication complexity of $g$ is
denoted by $N^0(g) = \log C^0(g)$. Finally, denote $C(g) =
C^0(g)+C^1(g)$ and $N(g) = \log C(g) \leq \max(N^0(g), N^1(g))+1$.
(An alternative to this combinatorial definition asks for the
number of bits that the parties need to exchange so as to verify
that $f(x,y)=b$.) All these definitions generalize naturally to
$k$-argument functions, where we consider combinatorial $k$-boxes
$B = X_1 \times \dots \times X_k$, rather than combinatorial
rectangles.

\myparagraph{Randomized Communication Complexity.}
For a function $\ggg$ and a positive number $0 \leq\epsilon <
\half$, denote by $R_\epsilon(g)$ the communication complexity of
the best randomized protocol for $g$ that errs on every input with
probability at most $\epsilon$, and denote $R(g) =
R_{\frac{1}{3}}(g)$. Newman~\cite{Newman91} proved that the
\emph{public-coin} model, where the players share a public random
string, is equivalent, up to an additive factor of $O(\log n)$
communication, to the \emph{private-coin} model, where each party
uses a private independent random string. Moreover, he proved that
w.l.o.g, the number of random strings used by the players in the
public-coin model is polynomial in $n$. All these results can be
easily extended to $k$-argument functions.
\begin{lemma}[\cite{Newman91}]\label{l:Newman}
There exist constants $c > 0, c'\ge 1$ such that for every Boolean function $\ggg$, if $R(g) = r(n)$ then there exists a protocol for $g$ in the
public-coin model with communication complexity $c'\cdot r(n)$
that uses random strings taken from a set of size $O(n^c)$.
\end{lemma}

\section{The Deterministic Case}
\label{s:det}
In this section we study the power of partition-argument lower bounds
in the deterministic case.

\begin{open}
\label{o:gapk} Let $k\ge 3$ be a constant integer and $f$ be a
$k$-argument function. What is the maximal gap between $D(f)$ and
the maximum $\max_{A,B} D(f^{A,B})$ over all partitions of $[k]$
into (disjoint) subsets $A$ and $B$?
\end{open}

In Section~\ref{ss:k}, we use multilinear algebra to show that under a
generalized version of the well known log-rank conjecture, partition
arguments are universal for multi-party communication complexity. We
also characterize the set of partitions one needs to study in order to
analyze the communication complexity of a $k$-argument function. Before
that, we give in Section~\ref{ss:3} a simpler proof for the case $k=3$.
This proof avoids the slightly more sophisticated multilinear algebra
needed for the general case.

\label{log-rank} Let $g:\bit^n \times \bit^n \ra \bit$ be a
Boolean two-argument function and $M_g \in \bit^{2^n \times 2^n}$
be the matrix representing it. It is well known that $\log
\rank(M_g)$ serves as a lower bound on the (two-party)
deterministic communication complexity of $g$.

\begin{theorem}[\cite{MS82}]
For any function $\ggg$, we have $D(g) \geq \log \rank(M_g).$
\end{theorem}

An important open problem in communication complexity is
whether the converse is true. This problem is known as the log-rank
conjecture. Formally,

\begin{conjecture}[Log Rank Conjecture]
\label{conj:org} There exists a constant $c \ge 1$ such that every
function $g:\bit^n \times \bit^n \ra \bit$ satisfies $D(g) =
O(\log^c \rank(M_g))$.
\end{conjecture}

It is known that if such a constant $c$ exists, then $c > 1 / 0.61
\approx 1.64$~\cite{NW95a}. As in the two-party case, in $k$-party
communication complexity still $\log \rank(M_f) \leq D(f)$; the
formal definition of $\rank(M_f)$ appears in Subsection~\ref{ss:k}
(and in Subsection~\ref{ss:3} for the special case $k=3$). This is
true for exactly the same reason as in the two-party case: any
deterministic protocol whose complexity is $c$ induces a partition
of the tensor $M_f$ into $2^c$ monochromatic $k$-boxes. Such boxes
are, in particular, rank-$1$ tensors whose sum is $M_f$. This, in
turn, leads to the following natural generalization of the above
conjecture.

\begin{conjecture}[Log Rank Conjecture for $k$-Party Computation]
\label{conj:k} Let $k$ be a constant. There exists a constant $c'
= c'(k) > 0$, such that for every function $\ffff$ we have that
$D(f) = O(\log^{c'} \rank(M_f))$.
\end{conjecture}

Computationally, even tensor rank in three dimensions is very
different than rank in two dimensions. While the former is
NP-Complete (see~\cite{Hastad90}), the latter can be computed very
efficiently using Gaussian elimination. However, in the
(combinatorial) context of communication complexity, much of the
properties are the same in two and three dimensions\remove{(see,
e.g., Lemma~\ref{c:Kron} in~Appendix~\ref{a:fs})}. We will show
below \remove{(in Appendix~\ref{a:further})} that, assuming
Conjecture~\ref{conj:k} is correct, the answer to
Question~\ref{o:gapk} is that the partition argument technique
always produces a bound that are polynomially related to the true
bound.
%% Furthermore, we characterize the collections of
%% partitions one has to analyze in order to understand the
%% complexity of the original $k$-argument function, assuming
%% Conjecture~\ref{conj:k} holds.

We start with the case $k=3$ whose proof is similar in nature to
the general case but is somewhat simpler and avoids the tensor
notation.

\subsection{The Three-Party Case}
\label{ss:3}
We start with the definition of a rank of three dimensinal matrices,
known as \emph{tensor rank}. In what follows $\F$ is any field.

\begin{definition}[Rank of a Three Dimensional Matrix]
A three dimensional matrix $M \in \F^{m \times m \times m}$ is of
rank $1$ if there exist three non-zero vectors $v,u,w \in \F^{m}$
such that, for every $x,y,z \in [m]$, we have that $M[x,y,z] =
v[x]u[y]w[z]$. In this case
we write $M = v \otimes u \otimes w$. A matrix $M\in \F^{m \times
m \times m}$ is of rank $r$ if it can be represented as a sum of
$r$ rank $1$ matrices (i.e., for some rank-1 three-dimensional
matrices $M_1,\ldots,M_r \in \F^{m \times m \times m}$ we have
$M=M_1+\ldots+M_r$), but cannot be represented as the sum of
$r-1$ rank $1$ matrices.
\end{definition}

The next theorem states that, assuming the log-rank conjecture for
$3$-party protocols, partition arguments are universal.
Furthermore, it is enough to study the communication complexity of
\emph{any two} of the three induced functions, in order to
understand the communication complexity of the original function.
We will use the notation $f^1 := f^{\set{1},\set{2,3}}$, $f^2 :=
f^{\set{2},\set{1,3}}$, and $f^3 := f^{\set{3},\set{1,2}}$.

\begin{theorem} \label{t:logRankPartition3party}
Let $f: \bit^n \times \bit^n \times \bit^n \ra \bit$ be a Boolean
function. Consider any two induced functions of $f$, say
$f^1,f^2$, and assume that Conjecture~\ref{conj:k} holds with
a constant $c'>0$. Then $D(f) = O((D(f^1)+D(f^2))^{c'})$.
\end{theorem}

Towards proving Theorem~\ref{t:logRankPartition3party}, we analyze
the connection between the rank of a three-dimensional matrix $M
\in \F^{m \times m \times m}$ and some related two-dimensional
matrices. More specifically, given $M$, consider the following
two-dimensional matrices $M_1, M_2, M_3 \in \F^{m \times m^2}$,
which we call the \emph{induced matrices} of $M$:
$$M_1[x, \pair{y,z}] = M[x,y,z], \ \ \ M_2[y, \pair{x,z}] = M[x,y,z], \ \ \ M_3[z, \pair{x,y}] = M[x,y,z]$$
We show that if $M$ has ``large'' rank, then at least two
of its induced matrices have large rank, as well
\footnote{It is possible to have
\emph{one} induced matrix with small rank. For example, if $M$ is
defined so that $M[x,y,z] = 1$ if $y=z$ and $M[x,y,z] = 0$
otherwise, then $M$ has rank $m$ while its induced matrix $M_1$ is
of rank $1$.}.

\begin{lemma}
\label{c:cuberank} Let $r_1 = \rank(M_1)$ and $r_2 = \rank(M_2)$.
Then $\rank(M) \leq r_1r_2$.
\end{lemma}
\begin{proof}
Let $v_1, \dots, v_{r_1} \in \F^n$ be a basis for the column space
of $M_1$. Let $u_1, \dots, u_{r_2}\in \F^n$ be a basis for the
column space of $M_2$. We claim that there are $r_1r_2$ vectors
$w_{1,1}, \dots, w_{r_1,r_2}$ such that $M =
\sum_{i=1}^{r_1}\sum_{j=1}^{r_2}v_i \otimes u_j \otimes w_{i,j}$.
This would imply that $\rank(M) \leq r_1r_2$, as required.

Fix $z \in [m]$ and consider the matrix $A_z\in \F^{m \times m}$ defined
by $A_z[x,y] = M[x,y,z]$. Observe that the columns of the matrix $A_z$
belong to the set of columns of the matrix $M_1$ (note that along each
column of $A_z$ only the $x$ coordinate changes, exactly as is the case
along the columns of the matrix $M_1$). Therefore, the columns of $A_z$
are contained in the span of $v_1, \dots v_{r_1}$. Similarly, the rows
of the matrix $A_z$ belong to the set of columns of the matrix $M_2$
(in each row of $A_z$, the value $x$ is fixed and $y$ is changed as is
the case along the columns of the matrix $M_2$) and are thus contained
in the span of vectors $u_1, \dots u_{r_2}$

Let $V \in \F^{m \times r_1}$ be the matrix whose columns are the
vectors $v_1, \dots v_{r_1}$. Similarly, let $U \in \F^{r_2
\times m}$ be the matrix whose rows are $u_1, \dots, u_{r_2}$. The
above arguments show that there exists a matrix $Q'_z \in \F^{m
\times r_2}$ such that $A_z = Q'_zU$ and $\rank(Q'_z) =
\rank(A_z)$. This is since the row space of $A_z$ contained in
the row space of $U$, and since the rows of $U$ are independent. Hence
the column space of the matrix $Q'_z$ is identical to the column
space of the matrix $A_z$, and so it is contained in the column space
of $V$. Therefore, there exists a matrix $Q_z
\in \F^{r_1 \times r_2}$ such that $Q'_z = VQ_z$. Altogether, we
get that $A_z = VQ_zU$. Simple linear algebraic manipulations show
that this means that $A_z = \sum_{i=1}^{r_1} \sum_{j=1}^{r_2}
Q_z[i,j]v_i \otimes u_j .$

Now, for every $i \in [r_1]$ and $j \in [r_2]$, define $w_{i,j}
\in \F^n$ such that, for every $z \in [m]$, we have that
$w_{i,j}[z] = Q_z[i,j]$. Then $M =
\sum_{i=1}^{r_1}\sum_{j=1}^{r_2}v_i \otimes u_j \otimes w_{i,j}.$
\end{proof}

\remove{The proofs of Theorem~\ref{t:logRankPartition3party} and
Corollary~\ref{cor:3}, based on the above Lemma, appear in
Appendix~\ref{a:3}.
}

\begin{proof} {\bf (of Theorem~\ref{t:logRankPartition3party})\ \ }
By the rank lower bound, $\log \rank (M_{f^1}) \leq D(f^1)$ and
$\log \rank (M_{f^2}) \leq D(f^2)$. By Lemma~\ref{c:cuberank},
$\rank(M_f) \le \rank(M_{f^1})\rank(M_{f^2})$. Therefore,
$$\log \rank(M_f) \le
\log \rank (M_{f^1}) + \log \rank (M_{f^2}) \leq D(f^1)+D(f^2).$$
Finally, assuming Conjecture~\ref{conj:k}, we get $D(f) =
O(\log^{c'} \rank(M_f))= O((D(f^1)+D(f^2))^{c'})$.
\end{proof}

\remove{
\begin{corollary}
\label{cor:3} Assume Conjecture~\ref{conj:org} is true. Then
Conjecture~\ref{conj:k} for $k=3$ is true if and only if partition
arguments are universal.
\end{corollary}

\begin{proof}
The first direction is given by
Theorem~\ref{t:logRankPartition3party}. For the reverse direction,
assume partition arguments are universal, and let $\fff$ be a
$3$-argument function. It is not hard to show that for all $i\in
[3]$, the rank of the matrix representing $f^i$ is at most the
rank of the matrix representing $f$, that is, $\rank(M_i) \leq
\rank(M_f)$. Now, by Conjecture~\ref{conj:org}, each induced
function $f^i$ of $f$ has communication complexity polynomial in
$\log \rank(M_i)$, and thus polynomial in $\log \rank(M_f)$. Since
we assume that partition arguments are universal, this implies
that the communication complexity of $f$ is polynomial in the
communication complexity of its induced functions, and thus is
polynomial in $\log \rank(M_f)$ as well, proving
Conjecture~\ref{conj:k} for $k=3$.
\end{proof}
}

%%% above there is an older corrolary with the same point.

\medskip

\begin{remark}
It is interesting to further explore the relations between the
following three statements:\\ (S1) partition arguments are
universal;\\ (S2) the (standard) 2-dimensional log rank conjecture
(Conjecture~\ref{conj:org}) holds; and \\ (S3) the 3-dimensional
rank conjecture holds.\\
%%% i.e. the case $k=3$ of ~\ref{conj:k})
Theorem~\ref{t:logRankPartition3party} shows that (S3) implies
(S1) and, trivially, (S3) implies (S2). We argue below, that (S1)
together with (S2) imply (S3). This implies that, assuming (S1),
the two versions of the rank conjecture, i.e. (S2) and (S3), are
equivalent. Similarly, it implies that, assuming (S2),
universality of partition arguments (S1) and the 3-dimensional rank
conjecture (S3) are equivalent.
It remains open whether the equivalence between the two
conjectures (S2) and (S3) can be proved, without making any
assumption.\\
To see that (S1) together with (S2) imply (S3), consider an
arbitrary 3-argument function $f$ of rank $r=\rank(M_f)$. Recall
that $f^1,f^2$ and $f^3$ denote the three induced functions of
$f$. It follows that, for $i\in[3]$, the (standard,
two-dimensional) rank of the matrix representing $f^i$ is bounded
by $r$. By (S2), for some constant $c$, we have $D(f^i) = O(\log^c
r)$, for $i\in[3]$. By (S1), for some constant $c'$, we have $D(f)
\le (\max\{D(f^1),D(f^2),D(f^3)\})^{c'} = O(\log^{c\cdot c'} r)$,
as needed.
\end{remark}

\subsection{The $k$-Party Case}
\label{ss:k}

We start with some mathematical background.

\myparagraph{Tensors, Flattening, Pairing, and Rank.}

Let $V_1,\ldots,V_k$ be vector spaces over the same field $\calF$; all
tensor products are understood to be over that field. For any subset
$I$ of $[k]$ write $V_I:=\bigotimes_{i \in I} V_i$. An element $T$
of $V_{[k]}$ is called a {\em $k$-tensor}, and can be written as a sum
of {\em pure} tensors $v_1 \otimes \cdots \otimes v_k$ where $v_i \in
V_i$. The minimal number of pure tensors in such an expression for $T$
is called the {\em rank} of $T$. Hence pure tensors have rank $1$.

If each $V_i$ is some $\calF^{n_i}$, then an element of the tensor
product can be thought of as a $k$-dimensional array of numbers
from $\calF$, of size $n_1 \times \cdots \times n_k$. A rank-1
tensor is an array whose $(j_1,\ldots,j_k)$-entry is the product
$a_{1,j_1} \cdots a_{k,j_k}$ where $(a_{i,j})_j$ is an element of
$\calF^{n_i}$.

%%%Enav - maybe add an explanation what tensors are...

For any partition $\{I_1,\ldots,I_m\}$ of $[k]$, we can view $T$
as an element of $\bigotimes_{l \in [m]} (V_{I_l})$; this is
called the {\em flattening} $\flat_{I_1,\ldots,I_m} T$ of $T$ or
just an $m$-flattening of $T$. It is the same tensor---or more
precisely, its image under a canonical isomorphism---but the
notion of rank changes: the rank of this $m$-flattening is the
rank of $T$ considered as an $m$-tensor in the space
$\bigotimes_{l \in [m]} U_l$, where $U_l$ happens to be the space
$V_{I_l}$.

If one views a $k$-tensor as a $k$-dimensional array of numbers, then
an $m$-flattening is an $m$-dimensional array. For instance, if $k=3$
and $n_1=2,n_2=3,n_3=5$, then the partition $\{\{1,2\},\{3\}\}$ gives
rise to the flattening where the $2 \times 3 \times
5$-array $T$ is
turned into a $6 \times 5$-matrix.

Another operation that we will use is {\em pairing}.  For a vector
space $U$, denote by $U^*$ the dual space of functions $\phi:U \ra
\calF$ that are $\calF$-linear, i.e., that satisfy
$\phi(u+v)=\phi(u)+\phi(v)$ and $\phi(cu)=c\phi(u)$ for all $u,v
\in U$ and $c \in \calF$.  Let $I$ be a subset of $[k]$, let
$\xi=\otimes_{i \in I} \xi_i \in \bigotimes_{i \in I}(V_i^*)$ be a
pure tensor, and let $T=\otimes_{i \in [k]} v_i \in V_{[k]}$ be a
pure tensor. Then the {\em pairing} $\lan T,\xi \ran \in V_{[k]
\setminus I}$ is defined as $ \lan T,\xi \ran= c \cdot \otimes_{i
\in [k] \setminus I} v_i ;$ where $c\in \calF$ is defined as $c :=
\left(\prod_{i \in I} \xi_i(v_i) \right)$. The pairing is
extended bilinearly in $\xi$ and $T$ to general tensors. Note that
$\xi$ induces a natural linear map $V_{[k]} \rightarrow V_{[k]
\setminus I}$, sending $T$ to the pairing $\lan T,\xi \ran$.

If one views a $k$-tensor as a $k$-dimensional array of numbers,
then pairing also reduces the dimension of the array. For instance,
pairing a tensor $T\in \calF^2 \otimes \calF^3 \otimes \calF^5$ with a
vector in the dual of the last factor $\calF^5$ gives a
linear combination of the five $2 \times 3$-matrices of which $T$
consists. Pairing with pure tensors corresponds to a repeated pairing
with dual vectors in individual factors.

Here are some elementary facts about tensors, rank,
flattening, and pairing:
\begin{description}
\item[Submultiplicativity] if $T$ is a $k$-tensor in
$\bigotimes_{i \in [k]} V_i$ and $S$ is an $l$-tensor in
$\bigotimes_{j \in [l]} W_j$, then the rank of the $(k+l)$-tensor
$T \otimes S$ is at most the product of the ranks of $T$ and $S$.
\item[Subadditivity] if $T_1,T_2$ are $k$-tensors in

$\bigotimes_{i \in [k]}V_i$, then the rank of the $k$-tensor
$T_1+T_2$ is at most the sum of the ranks of $T_1$ and $T_2$.
\item[Pairing with pure tensors does not increase rank] if $T \in
V_{[k]}$ and $\xi=\otimes_{i \in I}\xi_i$ then the rank of $\lan
T,\xi \ran$ is at most that of $T$. \item[Linear independence for
$2$-tensors] if a $2$-tensor $T$ in $V_1 \otimes V_2$ has rank
$d$, then in any expression $\sum_{p=1}^d R_p \otimes S_p=T$ with
$R_p \in V_1$ and $S_p \in V_2$ the set $\{S_1,\ldots,S_d\}$ is
linearly independent, and so is the set $\{R_1,\ldots,R_d\}$.
\end{description}

To state our theorem, we need the following definition.
\begin{definition}[Separating Collection of Partitions]
Let $k$ be a positive integer. Let $\mC$ be a collection of
partitions $\{I,J\}$ of $[k]=\{1,\ldots,k\}$ into two non-empty
parts. We say that $\mC$ is {\em separating} if, for every $i,j
\in [k]$ such that $i \neq j$, there exists a partition $\{I,J\}
\in \mC$ with $i \in I$ and $j \in J$.
\end{definition}

\begin{theorem} \label{t:logRankPartition}
Let $\ffff$ be a Boolean function. Let $\mC$ be a separating
collection of partitions of $[k]$ and assume that
Conjecture~\ref{conj:k} holds with a constant $c'>0$. Then
\[ D(f) = O(\ (2(k-1) \max_{\{I,J\} \in \mC}D(f^{I,J})\
)^{c'}\ ). \]
\end{theorem}

For a special separating collection of partitions we can
give the following better bound.

\begin{theorem}
\label{thm:Simple} Let $\ffff$ be a Boolean function. Set
$d_i:=D(f^{\{i\},[k]\setminus\{i\}})$ and assume that
Conjecture~\ref{conj:k} holds with a constant $c'>0$. Then $D(f)
= O((\sum_{i=1}^{k-1}d_i)^{c'})$.
\end{theorem}

These results will follow from upper bounds on the rank of $k$-tensors,
given upper bounds on the ranks of the $2$-flattenings corresponding
to $\mC$.

\begin{theorem} \label{thm:RankBound}
Let $V_1,\ldots,V_k$ be finite-dimensional vector spaces and let
$T$ be a tensor in their tensor product $\bigotimes_{i \in [k]}
V_i$. Let $\mC$ be a separating collection of partitions of $[k]$,
and let $\dmax$ be the maximal rank of any $2$-flattening
$\flat_{I,J} T$ with $\{I,J\} \in \mC$. Then $\rank T \leq
\dmax^{2(k-1)}$.
\end{theorem}

\begin{proof}
We prove the statement by induction on $k$. For $k=1$ the statement is
that $\rank T \leq 1$, which is true. Now suppose that $k>1$ and that
the result is true for all $l$-tensors with $l<k$ and all separating
collections of partitions of $[l]$. Pick $\{I,J\} \in \mC$ and write $
T=\sum_{p=1}^d R_p \otimes S_p $, where $R_p \in V_I$, $S_p \in V_J$,
$d \leq \dmax$, and the sets $R_1,\ldots,R_d$ and $S_1,\ldots,S_d$
are both linearly independent. This is possible by the condition that
the $2$-tensor (or matrix) $\flat_{I,J}T$ has rank at most $\dmax$.
As the $S_p$ are linearly independent, we can find {\em pure} tensors
$\zeta_1,\ldots,\zeta_d \in \bigotimes_{j \in J}(V_j^*)$ such that the
matrix $(\lan S_p, \zeta_q \ran)_{p,q}$ is invertible.

For each $q=1,\ldots,d$ set $ T_q:=\lan T,\zeta_q \ran \in V_I. $
By invertibility of the matrix $(\lan S_p, \zeta_q \ran)_{p,q}$
every $R_p$ is a linear combination of the $T_q$, so we can write
$T$ as $ T=\sum_{p=1}^d T_p \otimes S_p', $ where
$S_1',\ldots,S_d'$ are the linear combinations of the $S_i$ that
satisfy $\lan S_p',\zeta_q \ran=\delta_{p,q}$. Now we may apply
the induction hypothesis to each $T_q \in V_I$. Indeed, for every
$\{I',J'\} \in \mC$ such that $I \cap I', I\cap J' \neq
\emptyset$, we have $\flat_{I \cap I', I \cap J'} T_q=\lan
\flat_{I',J'} T, \zeta_q \ran, $ and since $\zeta_q$ is a pure
tensor, the rank of the right-hand side is at most that of
$\flat_{I',J'} T$, hence at most $\dmax$ by assumption. Moreover,
the collection
\[ \{\{I \cap I',I \cap J'\} \mid \{I',J'\} \in \mC \text{
with } I \cap I',I \cap J'\neq \emptyset\} \] is a separating
collection of partitions of $I$.  Hence each $T_q$ satisfies the
induction hypothesis and we conclude that $\rank T_q \leq
\dmax^{2(|I|-1)}$. A similar, albeit slightly asymmetric, argument
shows that $\rank S_q' \leq \dmax^{2(|J|-1)+1}$: there exist pure
tensors $\xi_1,\ldots,\xi_d \in \bigotimes_{i \in I}(V_i^*)$ such that the matrix $(\lan
T_p,\xi_r \ran)_{p,r}$ is invertible. This means that each $S_q'$
is a linear combination of the $d$ tensors
$ T_r':=\lan T,\xi_r \ran \in V_J,\ r=1,\ldots,d$.
The induction hypothesis applies to each of these, and hence $\rank S_q'
\leq \dmax \cdot \dmax^{2(|J|-1)}$ by subadditivity. Finally, using
submultiplicativity and subadditivity we find
\[ \rank T \leq \dmax \cdot \dmax^{2|I|-1} \cdot \dmax
\cdot \dmax^{2|J|-1} = \dmax^{2(k-1)}, \]
as needed.
\end{proof}

\begin{proof}{(of Theorem~\ref{t:logRankPartition})} By
Conjecture~\ref{conj:k}, we have $D(f)=O(\log^{c'} \rank(M_f))$.
Theorem~\ref{thm:RankBound} yields $\log \rank(M_f) \leq 2(k-1)
\max_{\{I,J\} \in \mC} \log \rank(M_{f^{I,J}})$, which by the rank
lower bound is at most
$2(k-1)\max_{\{I,J\} \in \mC} D(f^{I,J})$. This proves the
theorem.
\end{proof}

Just like Theorem~\ref{t:logRankPartition} follows from
Theorem~\ref{thm:RankBound}, Theorem~\ref{thm:Simple} follows immediately
from the following proposition.

\begin{proposition} \label{prop:RankBoundSimple}
Let $V_1,\ldots,V_k$ be finite-dimensional vector spaces and let
$T$ be a tensor in their tensor product $\bigotimes_{i \in [k]}
V_i$. Denote the rank of the $2$-flattening $\flat_{\{i\},[k]
\setminus \{i\}} T$ by $d_i$. Then $\rank T \leq d_1 \cdots
d_{k-1}$.
\end{proposition}

\begin{proof}
Denote by $U_i$ the subspace of $V_i$ consisting of all pairings
$\lan T,\xi \ran$ as $\xi$ runs over $V_{[k]\setminus \{i\}}^*$.
Then $\dim U_i=d_i$ and the $k$-tensor $T$ already lies in
$(\bigotimes_{i \in [k-1]}U_i) \otimes U_k$. After choosing a
basis of $\bigotimes_{i \in [k-1]}U_i$ consisting of pure tensors
$T_l\ (l=1,\ldots,d_1 \cdots d_{k-1})$, $T$ can be written (in a
unique manner) as $\sum_l T_l \otimes u_l$ for some vectors $u_l
\in U_k$. Hence $T$ has rank at most $d_1 \cdots d_{k-1}$.
\end{proof}
\remove{
\begin{corollary}
Assume Conjecture~\ref{conj:org} is true. Then
Conjecture~\ref{conj:k} is true if and only if {\em there exists}
a separating collection $\mC$ of partitions for which the
partition argument is universal, which in turn is equivalent to
the statement that {\em for every} separating collection $\mC$ the
partition argument is universal.
\end{corollary}

\begin{proof}
The first direction is given by Theorem~\ref{t:logRankPartition}.
For the reverse direction, let $\mC$ be any separating collection
for which the partition argument is universal, and let $\ffff$ be
a $k$-argument function. Then for all $\{I,J\} \in \mC$ the
flattening $M_{f^{I,J}}$ has rank at most $\rank M_f$.  By
Conjecture~\ref{conj:org}, each induced function $f^{I,J}$ of $f$
has communication complexity polynomial in $\log
\rank(M_{f^{I,J}})$, and thus polynomial in $\log \rank(M_f)$.
Since we assume that the partition argument with collection $\mC$
is universal, this implies that the communication complexity of
$f$ is polynomial in the communication complexity of its induced
functions, and thus is polynomial in $\log \rank(M_f)$ as well,
proving Conjecture~\ref{conj:k}.
\end{proof}
}

\begin{remark}
Theorem \ref{thm:RankBound} and Proposition \ref{prop:RankBoundSimple} are
special cases of the following more general rank bound, optimised relative
to the structure of $\mC$ and the individual bounds on flattenings.
Retain the notation of Theorem \ref{thm:RankBound}. For $\{I,J\}$ in the
separating collection $\mC$ let $d_{I,J}$ denote (an upper bound to)
the rank of $\flat_{I,J} T$. Recursively define a function $N$ from
non-empty subsets of $[k]$ to $\Nat$ as follows:
\[ N(H)=\begin{cases}
        1 \text{ if } |H|=1, \text{ and}\\
        \min\left[\{d_{I,J}^2 N(H \cap I)N(H \cap J) \mid
\{I,J\}
        \in \mC, H \cap I \neq \emptyset, H \cap J \neq
        \emptyset\} \cup \right.\\
        \left.
        \{d_{I,J} N(H \cap I)N(H \cap J)
        \mid \{I,J\} \in \mC, |H \cap I|=1,
        H\cap J \neq \emptyset\} \right] \text{ if } |H|\geq 2.
\end{cases}
\]
Then $\rank T \leq N([k])$. The proof is identical to that of
Theorem~\ref{thm:RankBound} with $[k]$ replaced by $H$, except
that if $|H \cap I|=1$, then one can choose the pure tensors
$\xi_r$ such that $\lan T_p,\xi_r \ran=\delta_{p,r}$. This implies
that $S'_q$ equals $\lan T,\xi_q \ran$, and one loses a factor
$d_{I,J}$. So for instance if $k=4$ and
$\mC=\{\{\{1,2\},\{3,4\}\},\{\{1,3\},\{2,4\}\}\}$ with bounds
$d_1$ and $d_2$, respectively, then we find the upper bound $d_1^2
d_2^2$.
\end{remark}

\section{Other Models of Communication Complexity}
\label{s:nd}
\subsection{The Nondeterministic Model}
\label{ssec:nd}

\remove{
In this section we show that in the nondeterministic model there
are functions with large communication complexity, for which all
induced functions have low communication complexity.}
As in the
deterministic case, the non-deterministic communication complexity
of the induced functions of a $k$-argument function $f$ gives a lower bound on the
non-deterministic communication complexity of $f$. It is natural
to ask the analogue of Question~\ref{o:gapk} for non-deterministic
communication complexity. We will show that the answer is negative:
there can be
an exponential gap between the non-deterministic communication
complexity of a function and that of its induced functions. Note that, for
proving the existence of a gap, it is enough to present such a gap
in the 3-party setting.

%%% remove first attempt
Not being able to find an explicit function $f$ for which
partition arguments result in lower bounds that are exponentially
weaker than the true non-deterministic communication complexity of
$f$, we turn to proving that such functions \emph{exist}. Towards
this goal, we use a well known combinatorial object---Latin
squares.

\begin{definition}[Latin square]
Let $m$ be an integer. A matrix $L \in [m]^{m \times m}$ is
a \emph{Latin square} of dimension $m$ if every row and every
column of $L$ is a permutation of $[m]$.
\end{definition}

The following lemma gives a lower bound on the number of Latin
squares of dimension $m$ (see, for example,~\cite[Chapter
17]{vLW92}).

\begin{lemma}
\label{l:latin} The number of Latin squares of dimension $m$ is at
least $\prod_{j=0}^{m}j!$. In particular, this is larger than
$2^{m^2/4}$.
\end{lemma}

Let $n$ be an integer and set $m = 2^n$. Let $L$ be a Latin square
of dimension $m$. Define the function $f_L: [m] \times [m] \times
[m] \ra \bit$ such that $f_L(x,y,z) = 1$ if and only if $L[x,y]
\neq z$. The non-deterministic communication complexity of
$f_L^1$, $f_L^2$ and $f_L^3$ is at most $\log n = \log \log m$. Indeed,
each of the induced functions locally reduces to the function
${\rm NE}_n:\bit^n \times \bit^n \ra \bit$, defined by ${\rm
NE}_n(a,b)=1$ iff $a \neq b$, for which it is known that $N^1({\rm
NE}_n) = \log n + 1$. For instance, for $f_L^1$, the player holding
$(y,z)$ locally computes the unique value $x_0$ such that
$L[x_0,y]=z$ and then the players verify that $x_0\neq x$. It is
left to prove that there exists a Latin square $L$ such that the
non-deterministic communication complexity of $f_L$ is
$\Omega(n)$.
%% (The following proof is non-constructive; it shows
%% the existence of such $L$ but does not provide a specific $L$.)
A simple counting yields the following lemma.

\begin{lemma}
The number of different covers of size $t$ of the $[m] \times [m]
\times [m]$ cube is at most $2^{3mt}$.
\end{lemma}

\begin{theorem}
\label{t:gapND} There exists a Latin square $L$ of dimension $m =
2^n$ such that the non-deterministic communication complexity of
$f_L$ is $n - O(1)$.
\end{theorem}
\begin{proof}
For two different Latin squares $L_0 \neq L_1$ of dimension $m$,
we have that $f_{L_0} \neq f_{L_1}$. In addition, no $1$-cover
$\calR$ corresponds to two distinct functions $f_{L_0}, f_{L_1}$.
%%% On the other hand, for every
%%% two different $1$-covers $\calR_0 \neq \calR_1$ of the $[m] \times
%%% [m] \times [m]$ cube, we have that $f_0 \neq f_1$ where
%%% $\calR_b$ is a $1$-cover of $f_b$, for $b \in \bit$.
Hence the number of covers needed to cover all the functions
$f_L$, where $L$ is a Latin square of dimension $m$, is at least
$2^{m^2 / 4}$. Let $t$ be the size of the largest $1$-cover among
this set of covers. Then we obtain $2^{3mt} \geq 2^{m^2 / 4}$.
Hence $3mt \geq m^2/4$, which implies $t \geq m / 12$. Therefore,
$\log t \geq \log m - \log 12 = n - \log 12$.
\end{proof}

\remove{
\begin{open}
Prove a similar separation for an explicit function.
\end{open}
}

\remove{
\begin{remark}
Here is a specific instantiation, $f$, of a Latin square function
for which proving that $N^1(f)=\Omega(n)$ would result in the desired
{\em explicit} separation. Unfortunately, we were not able to analyze
$N^1(f)$. Consider the function $\fff$ defined by $f(A,B,C) = 1$ iff
$m(A)m(B) \neq m(C)$, where $n$ is square and $m$ is any fixed surjective
map from $\F_2^n$ onto the set of invertible matrices in $\F_2^{\sqrt
n \times \sqrt n}$.  Again, all the induced functions of $f$ correspond
to ${\rm NE}_n$; Consider, for example, the induced function $f^1$. One
player gets $A$ and the other player gets $B$ and $C$. Now $AB \neq C$ iff
$A \neq CB^{-1}$. Since the second player can locally compute $B^{-1}$,
this reduces to computing ${\rm NE}_n(A, CB^{-1})$.
%%% , and the same holds for the other two induced functions.
\end{remark}
}

\remove{
\begin{remark}
We can use the standard ``list technique'' (see Chapter~4.1
in~\cite{KN97}) to prove a
quadratic separation between $N(f)$ and $\max(N(f^1), N(f^2),
N(f^3))$. Since $D(f) = O(N(f)^2)$, extending this gap also
implies the existence of a non-trivial gap for deterministic
communication complexity as well.
\end{remark}
}

\subsection{The Randomized Model}
\label{ssec:rnd}

Next, we show that partition arguments are also not sufficient for
proving tight lower bounds on the randomized communication
complexity.
%%% , as in the case of the non-deterministic model.
Let $f:\bit^n \times \bit^n \ra \bit$ be a Boolean function.
Recall that $R(f)$ denotes the communication complexity of a best
randomized protocol for $f$ that errs with probability at most
$1/3$. It is well known that $R({\rm NE}_n) = O(\log n)$. Again,
we use the functions defined by Latin squares of dimension $m =
2^n$. Our argument follows the, somewhat simpler,
non-deterministic case. On the one hand, as before, the three induced
functions are easily reduced to ${\rm NE}$ and hence their
randomized communication complexity is $O(\log n)$. To prove that
some of the functions $f_L$ are hard (i.e., an analog of
Theorem~\ref{t:gapND}), we need to count the number of distinct
randomized protocols of communication complexity $\log t$.

\begin{lemma}
\label{c:countRandom} The number of different randomized protocols
over inputs from $[m]\times[m]\times[m]$ of communication
complexity $r$ is $2^{m2^{O(r)}\poly(\log m)}$.
\end{lemma}

\begin{proof}
By Lemma~\ref{l:Newman}, any randomized protocol $\cal P$ with
communication complexity $r$ can be transformed into another
protocol ${\cal P}'$ with communication complexity $O(r)$ that
uses just $O(\log n)$ random bits, or, alternatively, $\poly(n) =
\poly(\log m)$ possible random tapes. Hence we can view any
randomized protocol of complexity $r$ as a set of $\poly(\log m)$
disjoint covers of the cube $[m]\times[m]\times[m]$, each
consisting of at most $2^{O(r)}$ boxes. The number of ways for
choosing each such box is $2^{3m}$ and so the total number of such
protocols is $2^{m2^{O(r)}\poly(\log m)}$.
\end{proof}

\begin{theorem}
\label{t:gapRand} There exists a Latin square $L$ of dimension $m
= 2^n$ such that the randomized communication complexity of $f_L$
is $\Omega(n)$.
\end{theorem}

\begin{proof}
By Lemma~\ref{l:latin}, the number of randomized protocols needed
to solve all the functions $f_L$ where $L$ is a Latin square of
dimension $m$ must be at least $2^{m^2 / 4}$---again, each
randomized protocol corresponds to at most one function, according
to the majority value for each input. Let $r$ be the maximum
randomized complexity of a function $f_L$ over the set of all
Latin squares $L$. Then we get that $2^{m2^{O(r)}\poly(\log m)}
\geq 2^{m^2 / 4}$. Hence $m2^{O(r)}\poly(\log m) \ge m^2/4$,
which implies $2^{O(r)} \ge m / \poly(\log m)$. Therefore, $r =
\Omega (\log m - \log\log m) = \Omega(n)$.
\end{proof}

\subsection{Deterministic Communication Complexity of Relations}
\label{ssec:rel}
In a communication protocol for a \emph{function}, Alice and Bob,
given inputs $x$ and $y$ respectively, have to compute a unique
value $f(x,y)$. In the more general setting of \emph{relations},
there is a set of values
%%% possibly the empty set (below and in [KW] not needed; for promise problems it is)
that are valid outputs for each input $(x,y)$. The study of
communication complexity of relations, beyond being a natural
extension that covers search problems and promise problems, is
important also for its strong implications to circuit
complexity~\cite{KW} (for a complete treatment see \cite[Chapter
5]{KN97}). Communication complexity of relations can be naturally
extended to more than two players. In this section, we show that
for some relations, partition arguments may only imply lower bounds
that are arbitrarily far from the true complexity of the relation.
This gives another example, where the communication complexity of
relations seems to behave differently than the communication
complexity of functions.

Let $f_1$, $f_2$, and $f_3$ be any two-argument functions whose
non-deterministic communication complexity is
$\Omega(n)$.\footnote{Many examples for such functions are known,
e.g. the function ${\rm IP}_n(x,y)$
%%% defined such that ${\rm IP}_n(x,y)=1$ iff $\pair{x,y}=1$
(inner product mod $2$).} For $x_1,x_2,y_1,y_2,z_1,z_2 \in
\bit^n$, let $x=x_1\conc x_2$, $y=y_1\conc y_2$, $z=z_1\conc z_2$
(the inputs to the 3-argument relation will be of length $2n$).
Define a relation $R \se \bit^{2n} \times \bit^{2n} \times
\bit^{2n} \times ([3]\times \bit)$ corresponding to $f_1,f_2$ and
$f_3$ such that $(x, y, z, (i,b))$ is in $R$ if one of the
following holds: (i) $i = 1$ and $f_1(x_1,y_1) = b$, or (ii) $i =
2$ and $f_2(x_2,z_1) = b$, or (iii) $i = 3$ and $f_3(y_2,z_2) =
b$.

\begin{observation}
\label{obs1} For every induced relation of $R$, it is easy to come
up with a correct output $(i,b)$ with no communication at all.
\end{observation}

\begin{lemma}
The deterministic communication complexity of the above
$3$-argument relation $R$ is $\Omega(n)$.
\end{lemma}

\begin{proof}
Let $P$ be a protocol of communication complexity $c$ for
computing $R$. That is, $P$ defines $2^c$ monochromatic boxes,
each labelled by some possible output; i.e., a pair $(i,b)$ where
$i \in [3]$ and $b \in \bit$. We will show that $c=\Omega(n)$
using the nondeterministic communication complexity of the
functions $f_1,f_2$ and $f_3$. Consider two following cases.\\
Case (i): for every $x_1,y_1 \in \bit^n$ there exist $x=x_1\conc
x_2$, $y=y_1\conc y_2$, and $z$
%%% =z_1\conc z_2$
such that $P(x,y,z) = (1, f_1(x_1,y_1))$. In this case, we claim
that $f_1$ has a non-deterministic protocol of complexity $c$. The
non-deterministic witness is a name of a rectangle in the protocol
$P$ that contains $(x,y,z)$ and is labelled by
$(1,f_1(x_1,y_1))$.\\
Case(ii): there exist $x_1,y_1 \in \bit^n$ such that for every
$x=x_1\conc x_2$, $y=y_1\conc y_2$, and $z=z_1\conc z_2$, either
$P(x,y,z) = (2, f_2(x_2,z_1))$ or $P(x,y,z) = (3, f_3(y_2,z_2))$.
Again, we split into to cases; Case (ii.a): for every $x_2, z_1
\in \bit^n$ there exist $z_2, y_2 \in \bit^n$ such that $P(x,y,z)
= (2, f_2(x_2,z_1))$. In this case, $f_2$ has a non-deterministic
protocol with complexity $c$, similarly to case (i). Case (ii.b):
there exist $x_2,z_1 \in \bit^n$, such that for every $z_2,y_2 \in
\bit^n$ we have that $P(x,y,z) = (3, f_3(y_2,z_2))$. In this case,
we get that $f_3$ has a deterministic protocol of complexity at
most $c$, which immediately implies it also has a
non-deterministic protocol of complexity at most $c$.
\end{proof}

\section{Fooling Set Arguments}
\label{a:fs} In Section~\ref{log-rank}, we proved that if the
log-rank conjecture is true, then any lower bound for $3$-argument
functions that can be proved using the rank lower bound method,
can also be proved using a partition argument. Moreover, if the
rank of the matrix representing a $3$-argument function is large,
then the rank of at least two of the matrices representing its
induced functions is large. In this section, we study the
situation for another popular lower bound method for communication
complexity, the fooling set method, and we show that the situation
here is very different. Namely, we show that there exist
$3$-argument functions for which a strong lower bound can be
proved using a large fooling set, while none of its induced
functions have a large fooling set. In fact, the gap is
exponential. This means that the fooling set technique may give,
in some cases, better lower bounds than what can be obtained by
using the partition argument and applying the fooling set method
to the induced functions. However, we also show that the fooling
set technique cannot yield lower bounds that are substantially
better than the rank lower bound. Recall the definition of fooling
sets for two-argument functions.

\begin{definition}[Fooling Set for $2$-Argument Functions]
Let $f:\bit^n \times \bit^n \ra \bit$ be a two-argument function.
A set of pairs $F = \set{(x_i,y_i)}_{i \in [t]}$ is called a
\emph{$b$-fooling set} (of size $t$) if: (i) for every $i \in
[t]$, we have that $f(x_i, y_i) = b$, and (ii) for every $i \neq j
\in [t]$, at least one of $f(x_i,y_j),f(x_j,y_i)$ equals $1-b$.
\end{definition}

To define a multi-party analogue, consider a boolean function
$f:(\bit^n)^k \to \bit$. For any pair $x,z \in (\bit^n)^k$ and any
partition $A,B$ of $[k]$ define the following ``mixture'' of $x$
and $z$, denoted $\select^{A,B}(x,z) \in (\bit^n)^k$, by
\[ (\select^{A,B}(x,z))_i:=
    \begin{cases}
        x_i \text{ if } i \in A \text{ and}\\
        z_i \text{ if } i \in B
    \end{cases}
\]
So for instance $\select^{\emptyset,[k]}(x,z)=z$ and
$\select^{[k],\emptyset}(x,z)=x$ and
$\select^{\{i\},[k]-\{i\}}(x,z)$ differs from $z$ at most in the
$i$-th position, where it equals $x_i$.

\begin{definition}[Fooling Set for $k$-Argument Functions]
Let $f:(\bit^n)^k \to \bit$ be a $k$-argument function and let $b
\in \bit$. A subset $F \se (\bit^n)^k$ is called a
\emph{$b$-fooling set} for $f$ if (i)~for all $x \in F$ we have
$f(x)=b$, and (ii)~for all pairs $x \neq z$ in $F$ the function
$f$ assumes the value $1-b$ on at least one element of the form
$\select^{A,B}(x,z)$ for some partition $A,B$ of $[k]$.
\end{definition}

Intuitively, the elements $\select^{A,B}(x,z)$ complement the
inputs $x$ and $z$ to a $2 \times \cdots \times 2$ $k$-dimensional
box. The fact that $f$ takes the value $1-b$ on at least one of
these elements implies that $x$ and $z$ cannot belong to the same
monochromatic box. This implies the following lemma, which is a
simple generalization of the fooling-set method from the two-party
case.

\begin{lemma}[\cite{Yao, LS81}]
If a function $f:(\bit^n)^k \to \bit$
%%%(respectively $\ggg$)
has a fooling set of size $t$ then $D(f) \geq \log t$.
%%% (resp., $D(g) \geq \log t$).
\end{lemma}

%%%%%CHECK THIS!

This subsection contains two results. First we show that a three-argument
function can have much larger fooling sets than any of its induced
two-argument functions. After that, we compare the fooling set lower
bound with the rank lower bound.

\begin{theorem}
\label{t:fooling} There exists a function $\fff$ such that $f$ has
a $1$-fooling set of size $2^n$ but no induced function of $f$ has
a fooling set of size $\omega(n)$.
\end{theorem}
\begin{proof}
The function is defined using the probabilistic method, i.e., we
look at some distribution on functions and prove that at least one
function in the support of this distribution satisfies the fooling
set requirements. The inputs $(x,y,z) \in \bit^n \times \bit^n
\times \bit^n$ are partitioned into three classes:
\begin{description}
\item[Three identical values.] If $x=y=z$, set $f(x,y,z) = 1$. We
later refer to these inputs as type~(a) inputs.
\item[Two identical values.] For every two distinct values $v_1,
v_2 \in \bit^n$, pick at random
%%% $r \in [6]$. According to this value, set
one of the six inputs $(v_1,v_1, v_2),(v_1,v_2, v_1),(v_2,v_1,
v_1),(v_1,v_2, v_2),(v_2,v_1, v_2)$ and $(v_2,v_2, v_1)$ and set
the value of $f$ on it to be $0$ and on the other five inputs to
be $1$.
\remove{   %%%% I thought that the table is not very useful
$$\begin{array}{ccccccc}
r & f(v_1,v_1, v_2) & f(v_1,v_2, v_1) & f(v_2,v_1, v_1) & f(v_1,v_2, v_2) & f(v_2,v_1, v_2) & f(v_2,v_2, v_1) \\
1 &        0        &         1       &         1       &         1       &         1       &         1       \\
2 &        1        &         0       &         1       &         1       &         1       &         1       \\
3 &        1        &         1       &         0       &         1       &         1       &         1       \\
4 &        1        &         1       &         1       &         0       &         1       &         1       \\
5 &        1        &         1       &         1       &         1       &         0       &         1       \\
6 &        1        &         1       &         1       &         1       &         1       &         0       \\
\end{array}$$
}
We later refer to these inputs as type~(b) inputs.
\item[Three distinct values.] For every $(x,y,z)$ such that $x$,
$y$ and $z$ are all distinct, pick at random $b \in \bit$ and set
$f(x,y,z) = b$. We later refer to these inputs as type~(c) inputs.
\end{description}

\begin{observation}
The function $f$, chosen as above, has a $1$-fooling set of size
$2^n$, with probability $1$.
\end{observation}
\begin{proof}
By the definition of $f$, the set $F = \set{(v,v,v):v \in \bit^n}$
is always a $1$-fooling set of size $2^n$. (Note that for this
claim we only rely on the inputs of types~(a) and~(b).)
\end{proof}

We proceed to show that, with positive probability (over the
choice of $f$), none of the induced functions of $f$ has a fooling
set of size $\omega(n)$. We analyze the probability that the
function $f^1(x,(y,z))$ has a fooling set of size $t = cn$, for
some constant $c>0$ to be set later, and show that it is smaller
than $\frac{1}{3}$. For symmetry reasons, the same analysis is
valid for the other two induced functions, and so the probability
that \emph{any} of them has a large fooling set is strictly
smaller than $1$, using a simple union bound.

Therefore, we focus on the induced function $f^1$. We prove that
the probability that a certain set $F$ of size $t$ is a fooling
set is extremely small. Then we multiply this probability by the
number of choices for $F$ and still get a probability smaller than
$\frac{1}{3}$.

\begin{observation}
\label{o:numberOfFs} The number of distinct choices of a set
$F=\set{(x_i, (y_i,z_i))}_{i \in [t]}$ is at most $2^{3nt}$.
\end{observation}

Let $F = \set{(x_i, (y_i,z_i))}_{i \in [t]}$ be a set of size $t$,
and $b\in \bit$. Consider the matrix $M_F \in \bit^{t \times t}$,
with rows labelled by $x_1, \dots, x_t$ and columns labelled by
$(y_1,z_1), \dots, (y_t,z_t)$. There are two types of columns in
$M_F$: (i) columns labelled by $(y,z)$ where $y=z$; and (ii)
columns labelled by $(y,z)$ where $y \neq z$. In every column of
type~(i), there is at most one entry that corresponds to an input
of type~(a), and all the rest correspond to inputs of type~(b). We
call the former a \emph{fixed} entry and the latter \emph{free}
entries. In every column of type~(ii) there are at most two
entries that correspond to inputs of type~(b) and the rest
correspond to inputs of type~(c). Again, we call the former
entries \emph{fixed} entries and the latter \emph{free} entries.
All together, out of the $t^2$ entries of the matrix $M_F$, there
are at most $2t$ fixed entries, and at least $t^2 - 2t$ free
entries.
\begin{observation}
For every $i \neq j \in [t]$, if both $M_F[x_i, (y_j,z_j)]$ and
$M_F[x_j, (y_i,z_i)]$ are free entries then $\Pr[f^1(x_i, (y_j,
z_j)) = b \mbox{ and } f^1(x_j, (y_i, z_i)) = b] \ge 1/36$.
\end{observation}

Note that the probability that two different pairs of inputs
satisfy the fooling set requirements are not independent because
of the manner in which we assigned the values of type~(b).
However, we can partition the entries into classes of size $6$,
such that every set of entries with at most one representative
from each class are independent. Hence we can pick $(t^2 - 2t) /
12$ pairs $i, j \in [t]$ such that the entries of $M_F$
corresponding to each of these pairs are set independently.
Therefore, the probability that the values assigned to all these
pairs respect the fooling set requirements is at most
$(\frac{35}{36})^{\frac{t^2 - 2t}{12}}$. The same analysis is
valid for the probability that $F$ is a $(1-b)$-fooling set. Thus,
setting the constant $c$ (where $t=cn$) such that $2^{3nt}
(\frac{35}{36})^{\frac{t^2 - 2t}{12}} < \frac{1}{6}$, we get that
there exists a function $f$ that satisfies the requirements of
Theorem~\ref{t:fooling}.
\end{proof}

Next, we show that the fooling set method cannot prove lower bounds
that are significantly stronger than the lower bounds proved for the
same function using the rank method. This extends a known result for the
two-party case~\cite{DHS96}, and strengthens the view that the behavior
of the rank method in the $k$-party case is similar to its behavior in
the two-party case.

\begin{theorem} \label{t:rankVersusFS}
Let $f:(\bit^n)^k \to \bit$ be a $k$-argument function, and assume that
$f$ has a fooling set of size $t$. Then $\rank(M_f) \geq t^{1/(2^k-2)}$.
\end{theorem}

The proof uses the following elementary lemma.

\begin{lemma} \label{l:Hadamard}
If $U$ and $V$ are $m \times m$-matrices over the field $\F$, then
the rank of their \emph{Hadamard product} $U \odot V$ defined by
$(U \odot V)[x,y]=U[x,y]V[x,y]$ is at most $\rank U \cdot \rank
V$. \qed
\end{lemma}

\begin{proof}
The $4$-tensor $U \otimes V$, which at position $[x,y,u,v]$ has
entry $U[x,y]V[u,v]$, has rank at most $\rank U \cdot \rank V$ by
submultiplicativity of the rank. Hence the same is true for its
$2$-flattening $\flat_{\{1,3\},\{2,4\}} U \otimes V$ corresponding to the partition $\{1,3\},\{2,4\}$,
which is an $m^2 \times m^2$-matrix with value $U[x,y]V[u,v]$ at
position $[[x,u],[y,v]]$. This flattening is known as the {\em
Kronecker product} of $U$ and $V$ and its rank is actually {\em
equal} to $\rank U \cdot \rank V$ for reasons that are irrelevant
here. Finally, the Hadamard product is the submatrix of the
Kronecker product corresponding to rows and columns labelled by
pairs of the form $(x,x)$ and $(y,y)$, respectively.
\end{proof}

\begin{proof}[of Theorem \ref{t:rankVersusFS}]
For each partition $A,B$ of $[k]$ consider the $t \times t$-matrix
$U^{A,B}$ whose rows and columns are labelled by elements of $F$ and
whose entry at position $[x,z]$ equals $f(\select^{A,B}(x,z))$. It
follows from the definition of $\select^{A,B}(x,z)$ that $U^{A,B}$ is
a submatrix of the flattening of $M_f$ corresponding to the partition
$A,B$ (perhaps up to repeated rows if several distinct elements of $F$
have the same $A$-parts, and similarly for columns). Hence we have
\[ \rank U^{A,B} \leq \rank M_f. \]
Now the Hadamard product of $U^{A,B}$ over all partitions $A,B$ of $[k]$
into two non-empty parts is the identity matrix---here we use that $F$
is a fooling set---and hence of rank $t$. Using Lemma \ref{l:Hadamard}
we find that
\[ (\rank M_f)^{2^k - 2} \geq t, \]
which proves the theorem.
\end{proof}

\myparagraph{Acknowledgments.} We wish to thank Ronald de Wolf for useful discussions
and, in particular, for referring us to~\cite{CKS03}.

\end{document}